\documentclass{jpsj3}

\usepackage{amsmath}
\usepackage{amssymb}
\usepackage{ascmac}
\usepackage{braket}
\usepackage{delarray}
\usepackage{here}

\usepackage{amsthm}
\theoremstyle{plain}
\newtheorem{theorem}{Theorem}[section]
\newtheorem*{theorem*}{Theorem}

\newtheorem*{definition*}{Definition}
\newtheorem{lemma}[theorem]{Lemma}
\newtheorem*{lemma*}{Lemma}
\newtheorem{corollary}[theorem]{Corollary}
\newtheorem{remark}[theorem]{Remark}

\usepackage{txfonts}

\usepackage{latexsym}

\usepackage{epsfig}
\usepackage{bm}
\usepackage{color}
\usepackage{ulem}

\newcommand{\be}{\begin{eqnarray}}
\newcommand{\ee}{\end{eqnarray}}
\newcommand{\ba}{\begin{array}}
\newcommand{\ea}{\end{array}}
\newcommand{\bmat}{\left(\begin{array}}
\newcommand{\emat}{\end{array}\right)}
\newcommand{\no}{\nonumber}

\title{Upper bound on the second derivative of the quenched pressure in spin-glass models: weak Griffiths second inequality}

\author{ Manaka Okuyama$^1$\thanks{manaka.okuyama.d2@tohoku.ac.jp} and Masayuki Ohzeki$^{1,2,3}$}
\inst{$^1$Graduate School of Information Sciences, Tohoku University, Sendai 980-8579, Japan \\
$^2$Institute of Innovative Research, Tokyo Institute of Technology, Yokohama 226-8503, Japan \\
$^3$Sigma-i Co., Ltd., Tokyo 108-0075, Japan} 

\abst{
The Griffiths first and second inequalities have played an important role in the analysis of ferromagnetic models.
In spin-glass models, although the counterpart of the Griffiths first inequality has been obtained, the counterpart of the Griffiths second inequality has not been established.
In this study, we generalize the method in the previous work [J. Phys. Soc. Jpn. {\bf76,} 074711 (2007)] to the case with multi variables  for both symmetric and non-symmetric distributions of the interactions, and derive some correlation inequalities for spin-glass models.
Furthermore, by combining the acquired equalities in symmetric distributions, we show that there is  a non-trivial positive upper bound on the second derivative of the quenched pressure with respect to the strength of the randomness, which is a weak result of the counterpart of the Griffiths second inequality in spin-glass models for general symmetric distributions.
}

\begin{document}
\maketitle

\section{Introduction}
The Griffiths inequalities, which characterize that the correlation of ferromagnetic models is always positive, are very important and indispensable in the rigorous analysis of ferromagnetic models.
Towards rigorous analyses of finite-dimensional spin-glass models, there are some previous studies~\cite{CG,CL,MNC,Kitatani,CS,CUV,CG2} which aim to establish the counterpart of the Griffiths inequalities in spin-glass models.
These attempts have been partially successful, and, when the distribution function of the interactions is symmetric, the counterpart of the Griffiths first inequality has been obtained in spin-glass models~\cite{CG,CL}.
Besides, it was shown that the ferromagnetic counterpart of the Griffiths second inequality~\cite{MNC,Kitatani} holds on the Nishimori line~\cite{Nishimori}.
However, no counterpart of the Griffiths second inequality has been established for other parameter regions~\cite{CUV}.
Rigorous analyses based on the concept of correlation inequalities have not been sufficiently advanced~\cite{CS,CG2}.

On the other hand, by focusing only on the distribution of a single interaction among all of the interactions, the previous studies~\cite{KNA,OO} showed that the quenched average of the local energy for Ising models with quenched randomness is always larger than or equal to one in the absence of all the other interactions. 
Although this is a non-trivial result that holds regardless of any other interaction, an extension to the case with multi variables was not apparent.
In this study, we give another derivation of their result~\cite{KNA,OO} and extend it to the case with multi variables.
Then, we obtain some correlation inequalities for the Ising models with quenched randomness.
Moreover, combining the acquired inequality and our previous work~\cite{OO2} for symmetric distributions, we find that there is a simple upper bound on the second derivative of the quenched pressure with respect to the strength of the randomness.
This bound can be regarded as a weak result of the counterpart of the Griffiths second inequality in spin-glass models for general symmetric distributions.

The organization of the paper is as follows.
In Sec. II, we define the model and explain the counterpart of the Griffiths inequalities in spin-glass models.
In Sec. III, we prove some inequalities for quenched averages in preparation for the next section.
Section IV is devoted to obtaining some correlation inequalities for spin glass, which is a systematic extension of the previous study~\cite{KNA,OO}.
Furthermore, we provide a non-trivial upper bound on the second derivative of the quenched pressure with respect to the strength of the randomness.
Finally, our conclusion is given in Sec. V.

\section{Ising model with quenched randomness and counterpart of Griffiths inequalities}

Following Ref. \cite{KNA}, we consider a generic form of the Ising model,
\be
H&=&- \sum_{A \subset{V}} \lambda_A  J_{A} \sigma_A ,
\\
\sigma_A &\equiv& \prod_{i\in A} \sigma_i ,
\ee
where $V$ is the set of sites, the sum over $A$ is over all the subsets of $V$ in which interactions exist, and the lattice structure adopts any form.
The probability distribution of a random interaction $J_A$ is represented as $P_A (J_A)$.
The probability distributions can be generally different from each other, i.e., $P_A (x)\neq P_{B} (x)$, and are also allowed to present no randomness, i.e., $P_A(J_A)=\delta(J-J_A)$.
The parameter $\lambda_A$ plays a role in controlling the strength of the randomness.

The partition function $Z_{\{J_{A}\}}$ and correlation function $\langle\sigma_B \rangle_{\{J_{A}\}}$ for a set of fixed interactions, $\{J_{A}\}$, are given by
\be
Z_{\{J_{A}\}}&=&\Tr \exp\left( \beta \sum_{A \subset{V}} \lambda_A  J_{A} \sigma_A \right),
\\
\langle\sigma_B \rangle_{\{J_{A}\}}&=& \frac{ \Tr \sigma_B \exp\left(\beta \sum_{A \subset{V}} \lambda_A  J_{A} \sigma_A \right)}{Z_{\{J_{A}\}}} .
\ee
The configurational average over the distribution of the interactions is written as 
\be
\mathbb{E}\left[g(\{J_{A}\})  \right] =  \left(\prod_{A \subset{V}} \int_{-\infty}^\infty dJ_A P_A(J_A) \right) g(\{J_{A}\}).
\ee
For example, the quenched average of the correlation function is obtained as
\be
\mathbb{E}\left[ \langle\sigma_B \rangle_{\{J_{A}\}} \right]&=&  \left(\prod_{A \subset{V}} \int_{-\infty}^\infty dJ_A P_A(J_A) \right)   \frac{ \Tr \sigma_B \exp\left(\beta \sum_{A \subset{V}}  \lambda_A J_{A} \sigma_A \right)}{Z_{\{J_{A}\}}} .
\ee
In addition, the quenched pressure $P$ is defined as
\be
P&=&\mathbb{E}\left[\log Z_{\{J_{A}\}}\right].
\ee

Recent studies showed~\cite{CG,CL}, when the probability distribution of a random interaction $J_B$ is symmetric, $P_B (J_B)=P_B (-J_B)$, the first derivative of  the quenched pressure with respect to the strength of the randomness is always positive:
\be
\frac{1}{\beta}\frac{\partial P}{\partial \lambda_B}&=&\mathbb{E}\left[ J_B\langle\sigma_B \rangle_{\{J_{A}\}} \right] \ge0,
\ee
which is regarded as the counterpart of the Griffiths first inequality in spin-glass models.

Next, we consider the counterpart of the Griffiths second inequality in spin-glass models.
Previous studies~\cite{CL,MNC,Kitatani,CUV} have focused on the second derivative of the quenched pressure with respect to the strength of the randomness,
\be
\frac{1}{\beta^2}\frac{\partial^2 P}{\partial \lambda_B \partial \lambda_C}&=&\mathbb{E}\left[ J_BJ_C (\langle\sigma_B \sigma_C \rangle_{\{J_{A}\}} -\langle\sigma_B  \rangle_{\{J_{A}\}} \langle \sigma_C \rangle_{\{J_{A}\}} ) \right] \label{second-deri}.
\ee
In ferromagnetic models, the Griffiths second inequality means that the second derivative of the pressure with respect to the ferromagnetic interactions is always positive.
Fortunately, on the Nishimori line, a similar relation also holds in spin-glass models and Eq.(\ref{second-deri}) is always positive~\cite{MNC,Kitatani}, which is the ferromagnetic counterpart of the Griffiths second inequality on the Nishimori line in spin-glass models.

In the case of symmetric distributions, however, studies on the counterpart of the Griffiths second inequality have not satisfactorily progressed.
When $P_B (J_B)$ and $P_C (J_C)$ follow the symmetric Gaussian distributions with the variance $\Lambda_B^2$ and $\Lambda_C^2$, respectively, by integration by parts, Eq. (\ref{second-deri}) is deformed to
\be
\frac{1}{\beta^2}\frac{\partial^2 P}{\partial \lambda_B \partial \lambda_C}
&=&\frac{\Lambda_B^2}{\beta} \frac{\partial}{\partial \lambda_C} \mathbb{E}\left[ -\langle\sigma_B \rangle_{\{J_{A}\}}^2 \right] .
\ee
This means that, if Eq. (\ref{second-deri}) is negative for $\sigma_B\neq \sigma_C$, the overlap expectation $\mathbb{E}\left[ \langle\sigma_B \rangle_{\{J_{A}\}}^2 \right]$ is monotonic non decreasing with the system size~\cite{CL}.
The overlap expectation $\mathbb{E}\left[ \langle\sigma_B \rangle_{\{J_{A}\}}^2 \right]$ tends to increase with increasing randomness.
Then, as the counterpart of the Griffiths second inequality, it is expected that Eq. (\ref{second-deri}) is always negative for general symmetric distributions; however, an explicit counterexample exists~\cite{CUV} and it was shown that Eq. (\ref{second-deri}) takes both positive and negative values depending on the details of the model.
Thus, the counterpart of the Griffiths second inequality has not been established even for symmetric distributions in spin-glass models and it is an important problem to investigate when Eq. (\ref{second-deri}) is negative.

On the other hand, it was shown that, for $P_B (J_B)=P_B (-J_B)$, the first derivative of  the quenched pressure with respect to the strength of the randomness has the following upper bound~\cite{KNA},
\be
\frac{1}{\beta}\frac{\partial P}{\partial \lambda_B}=\mathbb{E}\left[ J_B\langle\sigma_B \rangle_{\{J_{A}\}} \right] \le \mathbb{E}\left[ J_B \tanh(\beta \lambda_B J_{B}) \right]. \label{Kitatani-iq}
\ee
We note that Eq. (\ref{Kitatani-iq}) is independent of any other interaction, which is a non-trivial result.
The proof of Eq. (\ref{Kitatani-iq}) was obtained by focusing only on the distribution of a single interaction among all of the interactions; however, it is not clear how to extend it to the case with multi variables.
In the following sections, we give another proof of Eq. (\ref{Kitatani-iq}) and extend it to the case with multi variables.
Then, we obtain some correlation inequalities in the Ising models with quenched randomness.
Furthermore, in Sec. IV, we show that Eq. (\ref{second-deri}) has a non-trivial positive upper bound.

\section{Inequalities for expectations by inequality of arithmetic and geometric means
}
In this section, we prove three inequalities for expectations, which play an important role in the next section.

We consider the case that the distribution functions of $J_B$ and $J_C$ satisfy the following relations:
\be
P_B(-J_B)&=&\exp(-2\beta_{\text{NL}, B} J_B) P_B(J_B), \label{ferro-bias}
\\
P_C(-J_C)&=&\exp(-2\beta_{\text{NL}, C} J_C) P_C(J_C),\label{ferro-bias2}
\ee
where $\beta_{\text{NL}, B}$ and $\beta_{\text{NL}, C}$ are allowed to be any real values.
For example, in the case of the Gaussian distribution
\be
P_B(J_B)&=&\frac{1}{\sqrt{2\pi \mu^2}}\exp\left(-\frac{(J_B-J_0)^2}{2\mu^2} \right),
\ee 
and the binary distribution
\be
P_B(J_B)&=&p\delta(J_B-J)+(1-p)\delta(J_B+J) ,
\ee
$\beta_{\text{NL},B}$ is given as follows, respectively,
\be
\beta_{\text{NL}, B}&=& \frac{J_0}{\mu^2} ,
\\
\beta_{\text{NL}, B}&=&\frac{p}{1-p} .
\ee
We note that we do not impose any constraint on all the other interactions than $J_B$ and $J_C$.
In addition, for $f(\{J_A\})$, when we focus on the interaction $J_B$, we denote $f(\{J_A\})$ as $f(J_B, \{J_A\}\setminus J_B)$.
Similarly, when we are interested in $J_B$ and $J_C$, we represent $f(\{J_A\})$ as $f(J_B,J_C, \{J_A \}\setminus \{J_B,J_C\})$

Our first result in this section is as follows.
\begin{lemma}\label{lm1}
We assume that $X(\{J_A\})$ satisfies 
\be
&&X(\{J_A\}) > 0,
\\
&&X(-J_B, \{J_A \}\setminus J_B)=\frac{1}{X(J_B, \{J_A \}\setminus J_B)} .
\ee
Then, for any function $f(\{J_A\})$ satisfying
\be
&&f(\{J_A\}) \ge 0,
\\
&&f(J_B,\{J_A \}\setminus J_B)=f(-J_B,\{J_A \}\setminus J_B),
\ee
the following inequality holds
\be
\mathbb{E}\left[f(\{J_A\})X(\{J_A\}) \right] \ge \mathbb{E}\left[f(\{J_A\})  \exp(-\beta_{\text{NL}, B} J_B)  \right].
\ee
\end{lemma}

\begin{proof}
By dividing the integration interval of $J_B$ and summing up them, we obtain
\be
&&\mathbb{E}\left[ f(\{J_A\})X(\{J_A\}) \right]
\no\\
&=&\int_0^\infty dJ_BP_B(J_B)  \mathbb{E}\left[ f(\{J_A\})\left( X(J_B, \{J_A \}\setminus J_B)+X(-J_B, \{J_A \}\setminus J_B)e^{-2\beta_{\text{NL}, B} J_B}  \right) \right]' 
\no\\
&\ge&\int_0^\infty dJ_BP_B(J_B) \mathbb{E}\left[   2 f(\{J_A\})\sqrt{X(J_B, \{J_A \}\setminus J_B)X(-J_B, \{J_A \}\setminus J_B) e^{-2\beta_{\text{NL}, B} J_B}  }  \right]' 
\no\\
&=&\int_0^\infty dJ_BP_B(J_B) \mathbb{E}\left[   2 f(\{J_A\})e^{-\beta_{\text{NL}, B} J_B}    \right]' 
\no\\
&=&\mathbb{E}\left[f(\{J_A\}) e^{-\beta_{\text{NL}, B} J_B}  \right],
\ee
where $\mathbb{E}[\cdots]'$ denotes the configurational average over the randomness of the interactions other than $J_B$, and we used the inequality of arithmetic and geometric means.
Thus, we prove Lemma \ref{lm1}.
\end{proof}

Using Lemma \ref{lm1}, we give another derivation of the inequality for the local energy\cite{KNA,OO,OO2},
\be
\mathbb{E}\left[ - f(\{J_A\})J_B \langle \sigma_B \rangle_{\{J_{A}\}}  \right] \ge \mathbb{E}\left[-f(\{J_A\})J_B \left(\tanh(\beta \lambda_B J_{B})  + \frac{1-e^{-\beta_{\text{NL}, B} J_B}}{\sinh(2\beta \lambda_B J_{B})} \right) \right], \label{pre-kitatani}
\ee
where $f(\{J_A\})$ satisfies
\be
f(\{J_A\})&\ge&0,
\\
f(J_B,\{J_A \}\setminus J_B)&=&f(-J_B,\{J_A \}\setminus J_B).
\ee
We note that, for $\beta_{\text{NL}, B}=0$ and $f(\{J_A\})=1$, this inequality coincides with Eq. (\ref{Kitatani-iq}).
\begin{proof}
For simplicity, we denote $Z(J_B)$ as 
\be
Z(J_B)&=&\Tr \exp\left( \beta \sum_{A \subset{V}\setminus B} \lambda_A  J_{A} \sigma_A +\beta \lambda_B J_B \sigma_B \right),
\ee
We note that $Z(J_B)=Z_{\{J_{A}\}}$ but $Z(-J_B)\neq Z_{\{J_{A}\}}$.
Then, putting $X(\{J_A\})$ as $Z(-J_B)/Z(J_B)=\cosh(2\beta \lambda_B J_{B}) -  \langle\sigma_B \rangle_{\{J_{A}\}}  \sinh(2\beta \lambda_B J_{B})$ in Lemma. \ref{lm1}, we obtain
\be
\mathbb{E}\left[f(\{J_A\}) \left(\cosh(2\beta \lambda_B J_{B}) -  \langle\sigma_B \rangle_{\{J_{A}\}}  \sinh(2\beta \lambda_B J_{B})   \right)\right] \ge \mathbb{E}\left[f(\{J_A\}) e^{-\beta_{\text{NL}, B} J_B}\right]  .
\ee
Furthermore, replacing $f(\{J_A\})$ by $f(\{J_A\}){J_B}/{\sinh(2\beta J_{B})}$, we prove Eq. (\ref{pre-kitatani}).
\end{proof}

Our second result is a extension of Lemma. \ref{lm1} to two-variable case.

\begin{lemma}\label{lm2}
We assume that $X(\{J_A\})$ satisfies 
\be
&&X(\{J_A\}) >0,
\\
&&X(-J_B,-J_C, \{J_A \}\setminus \{J_B,J_C\})=\frac{1}{X(J_B,J_C, \{J_A \}\setminus \{J_B,J_C\})} .
\ee
Then, for any function $f(\{J_A\})$ satisfying
\be
f(\{J_A\})&\ge&0,
\\
f(J_B,J_C,\{J_A \}\setminus \{J_B,J_C\}))&=&f(-J_B,J_C,\{J_A \}\setminus \{J_B,J_C\}))
\no\\
&=&f(J_B,-J_C,\{J_A \}\setminus \{J_B,J_C\}))
\no\\
&=&f(-J_B,-J_C,\{J_A \}\setminus \{J_B,J_C\})),
\ee
the following inequality holds
\be
\mathbb{E}\left[f(\{J_A\})X(\{J_A\}) \right] \ge \mathbb{E}\left[f(\{J_A\})  \exp\left(-\beta_{\text{NL}, B} J_B\right)\exp \left(-\beta_{\text{NL}, C} J_C\right)  \right].
\ee
\end{lemma}

\begin{proof}
For simplicity, we denote $X(J_B,J_C, \{J_A \}\setminus \{J_B,J_C\})$ as $X(J_B,J_C)$.
By dividing the integration interval of $J_B$ and $J_C$ and summing up them, we find
\be
&&\mathbb{E}\left[ f(\{J_A\})X(\{J_A\}) \right]
\no\\
&=&\int_0^\infty dJ_BP_B(J_B)\int_0^\infty dJ_CP_C(J_C)  
\no\\
&&\mathbb{E}\left[ f(\{J_A\}) \left( X(J_B,J_C)+X(-J_B,-J_C)e^{-2\beta_{\text{NL}, B} J_B}e^{-2\beta_{\text{NL}, C} J_C} \right.\right.
\no\\
&&\left.\left. + X(-J_B,J_C)e^{-2\beta_{\text{NL}, B} J_B} +X(J_B,-J_C)e^{-2\beta_{\text{NL}, C} J_C}\right) \right]''
\no\\
&\ge&\int_0^\infty dJ_BP_B(J_B)\int_0^\infty dJ_CP_C(J_C)  \mathbb{E}\left[   4 f(\{J_A\})e^{-\beta_{\text{NL}, B} J_B}e^{-\beta_{\text{NL}, C} J_C} \right]''
\no\\
&=&\mathbb{E}\left[f(\{J_A\}) e^{-\beta_{\text{NL}, B} J_B} e^{-\beta_{\text{NL}, C} J_C} \right],
\ee
where $\mathbb{E}[\cdots]''$ denotes the configurational average over the randomness of the interactions other than $J_B$ and $J_C$, and we used the inequality of arithmetic and geometric means.
Therefore, we prove Lemma \ref{lm2}.

\end{proof}

Our third result is another extension of Lemma. \ref{lm1} to two-variable case.
\begin{lemma}\label{lm3}
We assume that $X(\{J_A\})$ and $Y(\{J_A\})$ satisfy
\be
X(\{J_A\})>0,
\\
Y(\{J_A\})>0,
\\
X(-J_B,J_C, \{J_A \}\setminus \{J_B,J_C\})&=&\frac{1}{X(J_B,J_C, \{J_A \}\setminus \{J_B,J_C\})} ,
\\
Y(J_B,-J_C, \{J_A \}\setminus \{J_B,J_C\})&=&\frac{1}{Y(J_B,J_C, \{J_A \}\setminus \{J_B,J_C\})} .
\ee
Then, for any function $f(\{J_A\})$ satisfying 
\be
f(\{J_A\})&\ge&0,
\\
f(J_B,J_C,\{J_A \}\setminus \{J_B,J_C\}))&=&f(-J_B,J_C,\{J_A \}\setminus \{J_B,J_C\}))
\no\\
&=&f(J_B,-J_C,\{J_A \}\setminus \{J_B,J_C\}))
\no\\
&=&f(-J_B,-J_C,\{J_A \}\setminus \{J_B,J_C\})),
\ee
the following inequality holds
\be
\mathbb{E}\left[f(\{J_A\})X(\{J_A\}) Y(\{J_A\}) \right] \ge \mathbb{E}\left[f(\{J_A\})   \exp\left(-\beta_{\text{NL}, B} J_B\right)\exp\left(-\beta_{\text{NL}, C} J_C\right)  \right].
\ee
\end{lemma}

\begin{proof}
For simplicity, we denote $X(J_B,J_C, \{J_A \}\setminus \{J_B,J_C\})$ and $Y(J_B,J_C, \{J_A \}\setminus \{J_B,J_C\})$ as $X(J_B,J_C)$ and $Y(J_B,J_C)$.
Similarly to the derivation of Lemma \ref{lm2}, we find
\be
&&\mathbb{E}\left[ f(\{J_A\})X(\{J_A\})Y(\{J_A\}) \right]
\no\\
&=&\int_0^\infty dJ_BP_B(J_B)  \mathbb{E}\left[ f(\{J_A\})\left(X(J_B,J_C)Y(J_B,J_C) +X(-J_B,J_C)Y(-J_B,J_C)e^{-2\beta_{\text{NL}, B} J_B}  \right) \right]' 
\no\\
&\ge&\int_0^\infty dJ_BP_B(J_B) \mathbb{E}\left[   2 f(\{J_A\})\sqrt{X(J_B,J_C) X(-J_B,J_C)Y(J_B,J_C) Y(-J_B,J_C)e^{-2\beta_{\text{NL}, B} J_B}  }  \right]' 
\no\\
&=&\int_0^\infty dJ_BP_B(J_B) \mathbb{E}\left[   2 f(\{J_A\})e^{-\beta_{\text{NL}, B} J_B}\sqrt{Y(J_B,J_C) Y(-J_B,J_C)  }  \right]' 
\no\\
&=&\int_0^\infty dJ_BP_B(J_B)  \int_0^\infty dJ_CP_C(J_C)  \mathbb{E}\left[   2 f(\{J_A\})e^{-\beta_{\text{NL}, B} J_B} \sqrt{Y(J_B,J_C) Y(-J_B,J_C) } \right.
\no\\
&&\left. + 2 f(\{J_A\})e^{-\beta_{\text{NL}, B} J_B} e^{-2\beta_{\text{NL}, C} J_C}\sqrt{Y(J_B,-J_C) Y(-J_B,-J_C)}  \right]''
\no\\
&\ge&\int_0^\infty dJ_BP_B(J_B)  \int_0^\infty dJ_CP_C(J_C)
\no\\
&&  \mathbb{E}\left[   4 f(\{J_A\})e^{-\beta_{\text{NL}, B} J_B} \sqrt{\sqrt{Y(J_B,J_C) Y(-J_B,J_C)  }  e^{-2\beta_{\text{NL}, C} J_C}\sqrt{Y(J_B,-J_C) Y(-J_B,-J_C) }}  \right]''
\no\\
&=&\int_0^\infty dJ_BP_B(J_B)  \int_0^\infty dJ_CP_C(J_C) \mathbb{E}\left[   4 f(\{J_A\})e^{-\beta_{\text{NL}, B} J_B} e^{-\beta_{\text{NL}, C} J_C}   \right]''
\no\\
&=&\mathbb{E}\left[f(\{J_A\})   \exp\left(-\beta_{\text{NL}, B} J_B\right)\exp\left(-\beta_{\text{NL}, C} J_C\right)  \right] ,
\ee
where we used the inequality of arithmetic and geometric means twice.
Thus, we arrive at Lemma \ref{lm3}.
\end{proof}

\section{ Some correlation inequalities in spin-glass models}
We have proved three lemmas for expectations in Sec III. 
In this section, using the acquired inequalities, we obtain some correlation inequalities in spin-glass models which is an extension of Eq. (\ref{pre-kitatani}) to the case with multi variables.

\subsection{Inequalities for double interactions}
Although the previous studies~\cite{KNA,OO} has focused only on a single interaction $J_B$ among all of the interactions, Lemma. \ref{lm2} enable us to extend their results to double interactions $J_B$ and $J_C$ among all of the interactions.
\begin{theorem} \label{th1}
Under the conditions (\ref{ferro-bias}) and (\ref{ferro-bias2}), the following inequality holds:
\be
&&\mathbb{E}\left[   J_B J_C\left( \langle\sigma_B   \sigma_C \rangle_{\{J_{A}\}}  +\frac{1-\exp\left(-\beta_{\text{NL}, B} J_B\right)\exp \left(-\beta_{\text{NL}, C} J_C\right)  }{\sinh(2\lambda_B\beta J_{B}) \sinh(2\lambda_C\beta J_{C}) }  \right) \right]
\no\\
&\ge&\mathbb{E}\left[ J_B J_C \left( \tanh(\beta \lambda_B J_{B}) \tanh(\beta \lambda_C J_{C})  +    \frac{\langle\sigma_B \rangle_{\{J_{A}\}} -\tanh(\beta \lambda_B J_{B})}{\tanh(2\lambda_C\beta J_{C})} +     \frac{\langle\sigma_C \rangle_{\{J_{A}\}} -\tanh(\beta \lambda_C J_{C})}{\tanh(2\lambda_B\beta J_{B})} \right)\right] . \label{main}
\no\\
\ee
\end{theorem} 
\begin{proof}
For simplicity, we denote $Z(J_B,J_C)$ as 
\be
Z(J_B,J_C)&=&\Tr \exp\left( \beta \sum_{A \subset{V}\setminus \{B,C\}} \lambda_A  J_{A} \sigma_A +\beta \lambda_B J_B \sigma_B + \beta \lambda_C J_C \sigma_C \right).
\ee
Setting $X(\{J_A\})$ as $Z(-J_B,-J_C)/Z(J_B,J_C)$ in Lemma. \ref{lm2}, we obtain
\be
&&\mathbb{E}\left[  f(\{J_A\}) \left( \cosh(2\beta J_{B})\cosh(2\beta J_{C})  + \sinh(2\beta J_{B})\sinh(2\beta J_{C}) \langle\sigma_B   \sigma_C \rangle_{\{J_{A}\}} 
\right. \right.
\no\\
&& \left.\left. - \sinh(2\beta J_{B})\cosh(2\beta J_{C}) \langle\sigma_B   \rangle_{\{J_{A}\}} -\cosh(2\beta J_{B})\sinh(2\beta J_{C}) \langle  \sigma_C \rangle_{\{J_{A}\}} \right) \right]
\no\\
&\ge&\mathbb{E}\left[ f(\{J_A\}) e^{-\beta_{\text{NL}, B} J_B} e^{-\beta_{\text{NL}, C} J_C}   \right].
\ee
Then, substituting $J_BJ_C/\left(\sinh(2\lambda_B\beta J_{B}) \sinh(2\lambda_C\beta J_{C})\right)$ into $f(\{J_A\})$, we arrive at
\be
&&\mathbb{E}\left[   J_B J_C\left( \langle\sigma_B   \sigma_C \rangle_{\{J_{A}\}}  +\frac{1-e^{-\beta_{\text{NL}, B} J_B} e^{-\beta_{\text{NL}, C} J_C}  }{\sinh(2\lambda_B\beta J_{B}) \sinh(2\lambda_C\beta J_{C}) }  \right) \right]
\no\\
&\ge&\mathbb{E}\left[ J_B J_C \left(\frac{1-\cosh(2\beta J_{B})\cosh(2\beta J_{C})  }{\sinh(2\beta J_{B})\sinh(2\beta J_{C})  } +    \frac{\langle\sigma_B \rangle_{\{J_{A}\}} }{\tanh(2\lambda_C\beta J_{C})} +     \frac{\langle\sigma_C \rangle_{\{J_{A}\}}}{\tanh(2\lambda_B\beta J_{B})} \right)\right] 
\no\\
&=&\mathbb{E}\left[ J_B J_C \left( \tanh(\beta \lambda_B J_{B}) \tanh(\beta \lambda_C J_{C})  +    \frac{\langle\sigma_B \rangle_{\{J_{A}\}} -\tanh(\beta \lambda_B J_{B})}{\tanh(2\lambda_C\beta J_{C})} +     \frac{\langle\sigma_C \rangle_{\{J_{A}\}} -\tanh(\beta \lambda_C J_{C})}{\tanh(2\lambda_B\beta J_{B})} \right)\right] .
\no\\
\ee
Thus, we obtain Eq. (\ref{main}).
\end{proof}
We mention that we had found the following inequality for $\beta_{\text{NL}, B}=\beta_{\text{NL}, C} =0$ in Ref. \cite{OO2},
\be
\mathbb{E}\left[   J_B J_C \langle\sigma_B   \sigma_C \rangle_{\{J_{A}\}}  \right]\ge0. \label{pre-0}
\ee
Numerical calculation suggests that the right-hand side in Eq. (\ref{main}) does not have a definite sign for $\beta_{\text{NL}, B}=\beta_{\text{NL}, C} =0$.
Thus, it is considered that Eq. (\ref{main}) is independent of Eq. (\ref{pre-0}).

Interestingly, an inequality of the same form as in Eq. (\ref{main}) holds for $\mathbb{E}\left[J_B J_C \langle\sigma_B  \rangle_{\{J_{A}\}} \langle  \sigma_C \rangle_{\{J_{A}\}}\right]$ as well.

\begin{theorem} \label{th2}
Under the conditions (\ref{ferro-bias}) and (\ref{ferro-bias2}), the following inequality holds:
\be
&&\mathbb{E}\left[   J_B J_C\left( \langle\sigma_B  \rangle_{\{J_{A}\}} \langle  \sigma_C \rangle_{\{J_{A}\}} +\frac{1-\exp\left(-\beta_{\text{NL}, B} J_B\right)\exp \left(-\beta_{\text{NL}, C} J_C\right)  }{\sinh(2\lambda_B\beta J_{B}) \sinh(2\lambda_C\beta J_{C}) }  \right)  \right]
\no\\
&\ge&\mathbb{E}\left[ J_B J_C \left( \tanh(\beta \lambda_B J_{B}) \tanh(\beta \lambda_C J_{C})  +    \frac{\langle\sigma_B \rangle_{\{J_{A}\}} -\tanh(\beta \lambda_B J_{B})}{\tanh(2\lambda_C\beta J_{C})} +     \frac{\langle\sigma_C \rangle_{\{J_{A}\}} -\tanh(\beta \lambda_C J_{C})}{\tanh(2\lambda_B\beta J_{B})} \right)\right] . \label{main2}
\no\\
\ee
\end{theorem} 
\begin{proof}
The proof is very similar to Theorem \ref{th1}.
We denote $Z(J_B)$ and $Z(J_C)$ as 
\be
Z(J_B)&=&\Tr \exp\left( \beta \sum_{A \subset{V}\setminus B} \lambda_A  J_{A} \sigma_A +\beta \lambda_B J_B \sigma_B \right),
\\
Z(J_C)&=&\Tr \exp\left( \beta \sum_{A \subset{V}\setminus C} \lambda_A  J_{A} \sigma_A +\beta \lambda_C J_C \sigma_C \right).
\ee
Putting $X(\{J_A\})$, $Y(\{J_A\})$ and $f(\{J_A\})$ as $Z(-J_B)/Z(J_B)$, $Z(-J_C)/Z(J_C)$ and $J_BJ_C/\left(\sinh(2\lambda_B\beta J_{B}) \sinh(2\lambda_C\beta J_{C})\right)$ in Lemma. \ref{lm3}, we obtain Eq. (\ref{main2}).
\end{proof}

\subsection{Upper bound on second derivative of quenched pressure}
We have extended the previous studies~\cite{KNA,OO} to two-variable cases.
However, at first glance, Eqs. (\ref{main}) and (\ref{main2}) take complex forms and it is not clear what these inequalities mean.
Here, we show that Eq. (\ref{main2}) enable us to obtain a simple upper bound on the second derivative of the quenched pressure with respect to the strength of the randomness. 

For $\beta_{\text{NL}, B}=\beta_{\text{NL}, C} =0$, we have already shown the following inequality~\cite{OO2},
\be
\mathbb{E}\left[   J_B J_C \tanh(\beta \lambda_B J_{B}) \tanh(\beta \lambda_C J_{C})    \right]&\ge&\mathbb{E}\left[   J_B J_C \langle\sigma_B   \sigma_C \rangle_{\{J_{A}\}}  \right] \label{pre-tanh}.
\ee
Then, combining Eqs. (\ref{main2}) and (\ref{pre-tanh}), we arrive at the following result.
\begin{corollary} 
For symmetric distribution $\beta_{\text{NL}, B}=\beta_{\text{NL}, C} =0$ and $\sigma_B\neq\sigma_C$, the second derivative of the quenched pressure with respect to the strength of the randomness has a non-trivial upper bound,
\be
\frac{1}{\beta^2}\frac{\partial^2}{\partial\lambda_B \partial\lambda_C }\mathbb{E}\left[  \log Z   \right] 
&=&\mathbb{E}\left[ J_BJ_C (\langle\sigma_B \sigma_C \rangle_{\{J_{A}\}} -\langle\sigma_B  \rangle_{\{J_{A}\}} \langle \sigma_C \rangle_{\{J_{A}\}} ) \right]
\no\\
&\le&\mathbb{E}\left[   J_B J_C \left( \frac{ \tanh(\beta \lambda_BJ_{B}) -\langle\sigma_B \rangle_{\{J_{A}\}}}{\tanh(2\beta \lambda_CJ_{C})} + \frac{\tanh(\beta \lambda_CJ_{C})-\langle\sigma_C \rangle_{\{J_{A}\}}}{\tanh(2\beta \lambda_BJ_{B}) }   \right) \right]\label{main3}.
\no\\
\ee
\end{corollary} 
\begin{remark} 
\textup{
We note that, from Eq. (\ref{pre-kitatani}), the right-hand side in Eq. (\ref{main3}) is always positive.
Thus, this bound can be considered as a weak result of the counterpart of the Griffiths second inequality for general symmetric distributions. }
\end{remark} 
In order to investigate the tightness of Eq. (\ref{main3}), for example, we consider a closed chain of six spins with one added interaction between $2$ and $5$,
\be
H&=&-\lambda_{1,2} J_{1,2} \sigma_1\sigma_2-\lambda_{2,3}J_{2,3} \sigma_2\sigma_3-\lambda_{3,4}J_{3,4} \sigma_3\sigma_4-\lambda_{4,5}J_{4,5} \sigma_4\sigma_5 
\no\\
&&-\lambda_{5,6}J_{5,6} \sigma_5\sigma_6-\lambda_{6,1}J_{6,1} \sigma_6\sigma_1-\lambda_{2,5}J_{2,5} \sigma_2\sigma_5, \label{example-model}
\ee
where all the interactions follows independently a symmetric binary distribution
\be
P_{}(J_{})=\frac{1}{2}\delta(J-1)+\frac{1}{2}\delta(J+1),
\ee
and $\lambda_{1,2}=\lambda_{2,3}=\lambda_{3,4}=\lambda_{4,5}=\lambda_{5,6}=\lambda_{2,5}=\lambda$.
This model had been investigated in the previous study~\cite{CUV} and, for $B=\{1,2\}$, $C=\{2,3\}$ and $\beta=1$, it was shown that the second derivative of the quenched pressure  takes a negative value for $\lambda<0.695$ and a positive value for $\lambda>0.695$. 
For this model (\ref{example-model}), we numerically calculate both sides in Eq. (\ref{main3}) for $B=\{1,2\}$, $C=\{2,3\}$ and $\beta=1$ from $\lambda=0.001$ to $\lambda=3$ in Fig. \ref{fig1}.
The numerical calculation shows that the acquired inequality is strict in the regions with small randomness ($\lambda<<1$) but give weak evaluation in the regions with strong randomness.
\begin{figure}[h]
\includegraphics{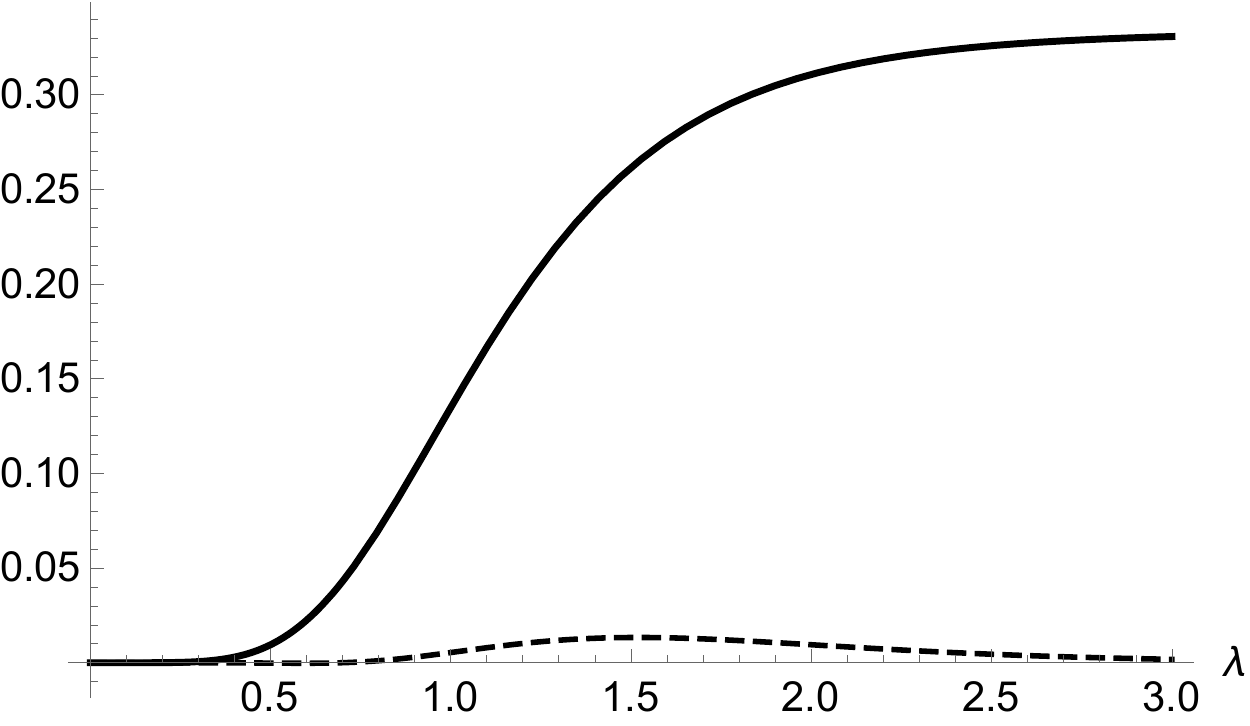}
\caption{Randomness dependence of both sides in Eq. (\ref{main3}) for $B=\{1,2\}$, $C=\{2,3\}$ and $\beta=1$ for the closed chain of six spins with one added interaction (\ref{example-model}).
The horizontal axis denotes the strength of the randomness, $\lambda$.
The dotted line and the thick line show the left-hand  side and the right-hand side in Eq. (\ref{main3}), respectively.}
\label{fig1}
\end{figure}

\section{ Conclusions}
We have obtained several correlation inequalities for the Ising models with quenched randomness.
Our main inequalities (\ref{main}) and (\ref{main2}) are extensions of previous studies~\cite{KNA,OO} (\ref{pre-kitatani}) to two-variable case.
Our method can be easily extended to the case with more than three variables.
Then, it is possible to obtain an infinite number of correlation inequalities in principle and the problem is how to find a meaningful one.

Furthermore, using the obtained inequalities for general symmetric distributions, we have given the positive upper bound (\ref{main3}) on the second derivative of the quenched pressure with respect to the strength of the randomness (\ref{second-deri}).
Numerical calculation shows that our bound (\ref{main3}) is strict in the regions with small randomness.
Equation (\ref{second-deri}) does not always take a negative value~\cite{CUV} and, thus, the counterpart of the Griffiths second inequality has not been established in spin glass models.
Our bound (\ref{main3}) is slight progress on this issue and can be regarded as a weak result of the counterpart of the Griffiths second inequality in spin glass models for general symmetric distributions, which is the first non-trivial upper bound on Eq. (\ref{second-deri}).

It is an important problem to improve our bound (\ref{main3}) to a tighter one.
One direction for future research is to consider model-dependent properties such as the shape of the lattice and the interaction.
Besides, it may also be useful to consider the effects of other interactions.
We have only focused on two interactions $J_B$ and $J_C$, and have not imposed any constraint on all the other interactions.
Incorporating information other than the two interactions may make our bound (\ref{main3}) tighter.

\begin{acknowledgment}
The present work was financially supported by JSPS KAKENHI Grant No. 18H03303, 19H01095, 19K23418, and the JST-CREST (No.JPMJCR1402) for Japan Science and Technology Agency.
\end{acknowledgment}


\end{document}